\documentclass [reqno, 12pt, a4paper] {amsart}

\usepackage{graphicx}
\usepackage{fancyhdr}
\usepackage{enumerate}
\usepackage{amsmath}
\usepackage{amsthm}
\usepackage{amssymb}
\usepackage{pdfsync}
\usepackage[colorlinks=true,linkcolor=blue,citecolor=magenta]{hyperref}

\usepackage{xcolor}

\usepackage{changebar}
\usepackage{pdfsync}

\makeindex
\newtheorem{theorem}{Theorem}[section]
\newtheorem{lemma}[theorem]{Lemma}
\newtheorem{proposition}[theorem]{Proposition}
\newtheorem{corollary}[theorem]{Corollary}

\newcommand{\mc}[1]{{\mathcal #1}}
\newcommand{\mf}[1]{{\mathfrak #1}}
\newcommand{\mb}[1]{{\mathbf #1}}
\newcommand{\bb}[1]{{\mathbb #1}}
\newcommand{\bs}[1]{{\boldsymbol #1}}

\newcommand\R{{\mathbb R}}

\newcommand\Z{{\mathbb Z}}

\newcommand\ve{\varepsilon}

\newcommand\ds{\displaystyle}

\numberwithin{equation}{section}

\begin{document}

\title{Microscopic derivation of an adiabatic thermodynamic transformation}

\author{Stefano Olla}
  \address{CEREMADE, UMR CNRS 7534\\
  Universit\'e Paris-Dauphine\\
  75775 Paris-Cedex 16, France \\
  \texttt{{ olla@ceremade.dauphine.fr}}}

\author{Marielle Simon}

\address{Departamento de Matem\' atica, PUC-RIO, Rua Marqu\^es de S\~ao Vicente, no. 225, 22453-900, Rio de Janeiro, Brazil,  \texttt{{ marielle.simon@mat.puc-rio.br}}}
\address{UMPA ENS Lyon, UMR CNRS 5669, 46 all\'ee d'Italie, 69007 Lyon, FRANCE.}

%\date{22sep2014 {\bf File: {\jobname}.tex.}} 
% \dedicatory{To Errico Presutti}

\begin{abstract}
We obtain macroscopic adiabatic thermodynamic transformations by
space-time scalings of a microscopic Hamiltonian dynamics subject to
random collisions with the environment.
 The microscopic dynamics is given by a chain of 
oscillators subject to a varying tension (external force) and 
 to collisions with external independent particles of ``infinite mass''. 
The effect of each collision is to change the sign of the
velocity without changing the modulus. This way the energy is conserved by
the resulting dynamics.   
 After a diffusive space-time scaling and
coarse-graining, the profiles of volume and energy converge to the solution of a
deterministic diffusive system of equations with boundary conditions
given by the applied tension. This defines an irreversible
thermodynamic transformation from 
an initial equilibrium to a new equilibrium given by the final tension
applied. Quasi-static reversible adiabatic transformations are then
obtained by a further time scaling. Then we
prove that the relations between the limit work, internal energy and
thermodynamic entropy agree with the first and second principle of
thermodynamics.  

\end{abstract}
 \thanks{This work has been partially supported by the
  European Advanced Grant {\em Macroscopic Laws and Dynamical Systems}
  (MALADY) (ERC AdG 246953), by the fellowship L'Oreal
  France-UNESCO \emph{Pour les femmes et la science}, and by the CAPES
    and CNPq program \emph{Science Without Borders}.} 
% \keywords{...}
% \subjclass[2000]{...}

\maketitle

%%%%%%%%%%%%%%%%%%%%%%%%%%%%%%%%%%%%%%%%%%%%%%%%
%%%%%%%%%%%%%%%%%%%%%%%%%%%%%%%%%%%%%%%%%%%%%%%%
\section{{Introduction}\label{sec:intro}}
%%%%%%%%%%%%%%%%%%%%%%%%%%%%%%%%%%%%%%%%%%%%%%%%
%%%%%%%%%%%%%%%%%%%%%%%%%%%%%%%%%%%%%%%%%%%%%%%%

In classical thermodynamics, adiabatic transformations are defined as
those processes that change the state of the system from an equilibrium to another
only by the action of an external force. This means that the system is
isolated, not in contact with any \emph{heat bath}, and that the
change in its internal energy $U$ is only due to the work done by the
applied external force. The second law of thermodynamics states that the only possible
adiabatic transformations are those that do not decrease the
thermodynamic entropy $S$ of the system. Irreversible adiabatic
transformations assume a strict increase
of the entropy, while if entropy remains constant the transformation
is called reversible or quasi-static. 

When connecting this transformation to the microscopic dynamics of the
atoms constituting the system, 
we understand this thermodynamic behaviour as the macroscopic
\emph{deterministic} change of the observables that characterize the
thermodynamic equilibria (in the case studied in this article, the
energy and the volume, or the temperature and the tension). We intend
\emph{macroscopic} in the sense that we would like to recover this
behaviour in a large space and time scale: the thermodynamic system is
composed by a huge number of atoms and we look at a very large time
scale with respect to the typical frequency of atoms vibration.
Mathematically this means a space-time scaling limit procedure.

We study these adiabatic transformations in a one dimensional model of
a wire. Macroscopically the equilibrium states are characterized by
the length $L$ and the energy $U$ (as extensive quantities), or by the
temperature $T= \beta^{-1}$ and the tension $\tau$. Microscopically we
model this wire by a Hamiltonian system constituted by a chain of
springs attached at one extreme to a point, while at the other extreme a force
$\bar\tau$ acts on the last particle. The Hamiltonian dynamics of the
chain is perturbed by independent random changes of the
sign of velocities. This random perturbation can be
seen as the effect of collisions with \emph{environment} particles of infinite mass
moving independently, in orthogonal direction to the wire. Notice that
these random collisions conserve the energy of particles, so that the
dynamics is still adiabatic. 

The first effect of these random perturbations is to ensure that
the only parameters characterizing the macroscopic equilibrium states are the
energy and the length, i.e. that the system obeys the so called
0th law of thermodynamics. In fact these random perturbations select
the Gibbs probability 
measures on the configurations,
parametrized by the conserved quantities, as the only stationary measures for
the corresponding infinite dynamics (for details see \cite{ffl,stefano}).

Another important consequence of these collisions is the suppression
of momentum conservation, so that there is no ballistic transport on a 
macroscopic scale. Thus, we expect a diffusive behaviour of the energy
and the volume stretch caused by a change of the exterior tension
$\bar\tau$,  before attaining the new equilibrium. Consequently the
correct space-time macroscopic rescaling is
diffusive. The change of the external force $\bar\tau$ should happen
on the macroscopic time scale, i.e.~ very slowly with respect to the
typical time scale of the dynamics of the atoms. 

We expect that, under a diffusive space-time scale, the empirical
profiles of the stretch and the energy, due to a change of the
applied tension $\bar\tau$, evolve deterministically following the
diffusive system of partial differential equations \eqref{eq:diff}.
The solution of this system eventually will converge to a new
equilibrium state.  
This deterministic evolution of the profiles describes an irreversible    
adiabatic trasformation, and, as shown in \autoref{sec:therm-cons},
it increases the thermodynamic entropy of the system. The reversible or
quasi-static transformations are then obtained by a further rescaling
of time, see \autoref{sec:quasistatic-limit}, similar as proposed in
\cite{bertini-prl,bertini-jsp,olla}. It should be possible to obtain
these quasi-static transformation in a direct limit at a larger
(subdiffusive) time scale, this will be object of further investigation.

The scaling limit for the non-linear system is still out of the known mathematical
techniques, as it requires to deal with the \emph{non-gradient}
energy current in the energy conservation law. Even though the
convergence of the Green-Kubo formula defining the energy diffusivity
is proven in \cite{bo2}, the actual proof of the macroscopic equation
requires a fluctuation-dissipation decomposition of the energy current
(cf. \cite{os} for such decomposition in a non-linear dynamics
conserving only energy). In the linear case (harmonic oscillators),
there is an explicit fluctuation-dissipation decomposition of the
energy current and it is possible to perform the scaling limit. This
was done in \cite{simon} for the periodic boundary conditions case. We
adapt here that proof for the case of  mixed boundary conditions with
slowly changing external tension. 

In \cite{oe}, the macroscopic limit was studied in the same model, for
non-linear springs, but
with a stochastic exchange of momentum between nearest neighbour
particles. This dynamics also conserves the momentum, besides the
energy and the volume. For that system the macroscopic space-time scale is 
hyperbolic, and the macroscopic equations are given by the Euler
system of conservation laws. Notice that in the harmonic case these are just
linear wave equations, and the corresponding macroscopic equation will
not bring the system to a new equilibrium state, that can be reached
only at a super-diffusive space-time scale \cite{jko}. 
In the non-linear case we need a better understanding of the entropy
production of the shock waves that appear in the solution to
Euler equations.
 
Isothermal transformations in this model have been deduced in
\cite{olla} in the non-linear case, where the heat bath is modelled by
Langevin thermostats. In this evolution only the volume evolves
macroscopically. In \cite{olla} these heat baths act on the bulk
of the chain, at every point. If we want to make them act only at the
boundaries of the chain, then we should obtain the same macroscopic
equations as in the present article, but with boundary conditions
corresponding to the thermostat temperature (this will be object of
further investigation).

With the result contained in the present article 
we complete the deduction of the macroscopic Carnot cycle from the
microscopic dynamics.

%%%%%%%%%%%%%%%%%%%%%%%%%%%%%%%%%%%%%%%%%%%%%%%%%%%
 %%%%%%%%%%%%%%%%%%%%%%%%%%%%%%%%%%%%%%%%%%%%%%%%%%%
%%%%%%%%%%%%%%%%%%%%%%%%%%%%%%%%%%%%%%%%%%%%%%%%%%%

\section{Adiabatic microscopic dynamics}
\label{sec:adiab-micr-dynam}

We consider a chain of $n$ coupled oscillators in one dimension.
Each particle has the same mass that we set equal to 1.
The position of atom $i$ is denoted by $q_i\in \mathbb R$, while its
momentum is denoted by $p_i\in\mathbb R$. 
 Thus the configuration space is $(\mathbb R\times \mathbb R)^n$. 
We assume that an extra particle $0$ is attached to a fixed point
and does not move, i.e. $(q_0,p_0)\equiv(0,0)$,
 while on particle $n$ we apply a force $\bar\tau(t)$ depending on time.
Observe that only the particle 0 is constrained to not move, and that
$q_i$ can assume also negative values. 

Denote  ${\bf q} :=(q_1,\dots,q_n)$ and ${\bf p}
:=(p_1,\dots,p_n)$. The interaction between two particles $i$ and
$i-1$ is described by the potential energy $V(q_i-q_{i-1})$ of an
anharmonic spring relying the particles. We assume   $V(r)$ to be a
positive smooth function which for large $r$ grows faster than linear
but at most quadratic, that means that there exists a constant $C>0$
such that 
\begin{align*}
&\lim_{|r|\rightarrow\infty}\frac{V(r)}{|r|}=\infty,\\
~\\
&\limsup_{|r|\rightarrow\infty}V^{\prime\prime}(r)\leqslant C<\infty.\end{align*}
Energy is defined by the following Hamiltonian:
\[
\sum_{i=1}^{n}  \left( \frac{p_i^2}2 + V(q_{i}-q_{i-1}) \right).
\]
Since we focus on a nearest neighbor interaction, we may define the
distance between particles by \index{{\large{CHAPTER 2:}}!$r_i$} 
$$
r_i=q_{i}-q_{i-1}, \qquad i=1,\dots,n.
$$ 
The particles are subject to an interaction with the environment that
does not change the energy: each particle has an independent
Poissonian clock and its momentum changes sign when it rings. The equations of motion are given by
\[ \left\{ \begin{aligned}
    dr_i(t) &= n^2 \big(p_i(t) - p_{i-1}(t)\big)\; dt\\
    dp_i(t) &= n^2 \big(V'(r_{i+1}(t)) - V'(r_i(t))\big)\; dt - 2p_i(t^-)
    \; d\mathcal N_i(\gamma n^2t), \quad i=1,\dots, n-1,\\
    dp_n(t) &= n^2\big(\bar\tau(t) - V'(r_n(t))\big)\; dt - 2p_n(t^-)
   \;  d\mathcal N_n(\gamma n^2t). 
  \end{aligned}\right.\]
Here $\{\mathcal N_i(t)\}_i$ are n-independent Poisson processes of
intensity 1, the constant $\gamma$ is strictly positive, and $p_0$ is set identically to $0$. We have 
already rescaled time according to the diffusive space-time scaling. Notice that
$\bar\tau(t)$ changes at this macroscopic time scale. The generator of this diffusion is given by
\begin{equation*}
\mathcal L_n^{\bar\tau(t)}:= n^2 A^{\bar\tau(t)}_n + n^2 \gamma S_n.
\end{equation*}
Here the Liouville operator  $A^{\tau}_n$ is given by
\[ A^{\tau}_n
 =\sum_{i=1}^{n}\big(p_{i}-p_{i-1}\big)
 \frac{\partial}{\partial  r_i}+
 \sum_{i=1}^{n-1}\big(V^{\prime}(r_{i+1})-V^{\prime}(r_{i})\big)
 \frac{\partial}{\partial p_i}+\big(\tau -V^{\prime}(r_n)\big)\frac{\partial}{\partial p_n},
 \]
while, for $f : (\R \times \R)^n \to \R$,
\begin{equation*}
  S_n f (\mb r, \mb p) = \sum_{i=1}^n \left( f(\mb r, \mb p^i) -
    f(\mb r,\mb p)\right) 
\end{equation*}
where $(\mb p^i)_j = p_j$ if $j\neq i$ and  $(\mb p^i)_i = -p_i$. For $\bar\tau(t) = \tau$ constant, the system has a family of
stationary measures given by the canonical Gibbs distributions
\begin{equation}
  \label{eq:gibbs}
  d\mu^n_{\tau,T} = \prod_{i=1}^n e^{- \frac 1T (\mc E_i - \tau
      r_i) - \mc G_{\tau,T}}\; dr_i\; dp_i, \qquad  T>0, 
\end{equation}
where we denote 
$$
\mc E_i = \frac{p_i^2}{2} + V(r_i),
$$
the energy that we attribute to the particle $i$, and 
\begin{equation}
  \label{eq:pfunct}
  \mc G_{\tau,T} = \log \left[\sqrt{2\pi T}\int e^{-\frac 1T
      (V(r) - \tau r)}\; dr  \right].
\end{equation}
Observe that the function 
$\mf r(\tau,T) = T \partial_\tau \mc G_{\tau,T}$ gives the average equilibrium length in function of
the tension $\tau$, and  \[\mf u(\tau,T) = \tau \mf r(\tau,T)  + T^2\partial_T \mc G_{\tau,T}\] 
is the corresponding thermodynamic internal energy function. We denote the inverse of the average length $\mf r$ by $\bs \tau(\mf r,\mf u)$. Thermodynamic entropy $S(\mf r, \mf u)$ is defined as
\begin{equation}
  \label{eq:S}
  S(\mf r,\mf u) = \frac 1T \left( \mf u - \bs\tau  \mf r\right)  +  \mc G_{\bs\tau,T}
\end{equation}
so that $\partial_{\mf u} S = T^{-1}$ and $\partial_{\mf r} S = -T^{-1}\bs\tau$. 
From now on, we reindex notations by using the inverse temperature $\beta:=T^{-1}$. In the following we will need to consider local Gibbs measures (non homogeneous
product), corresponding to profiles of tension and temperature
$\{\tau(x),\beta^{-1}(x), x\in[0,1]\}$:
\begin{equation}
  \label{eq:gibbs}
  d\mu^n_{\tau(\cdot),\beta(\cdot)} = \prod_{i=1}^n e^{-\beta(i/n) \big(\mc E_i - \tau(i/n)
      r_i\big) - \mc G_{\tau(i/n),\beta(i/n)}}\; dr_i\; dp_i. 
\end{equation}
Given an initial profile of tension $\tau(0,x)$ and temperature $\beta^{-1}(0,x)$,
we assume that the initial probability state is given by the corresponding
$\mu^n_{\tau(0,\cdot),\beta(0,\cdot)}$.  This implies the following
convergence in probability with respect to 
the initial distribution: 
\begin{equation}
  \label{eq:4}
  \begin{split}
    \frac 1n\sum_{i=1}^n G(i/n) r_i(0) \longrightarrow \int_0^1 G(x)
    \mf r(\tau(0,x),\beta(0,x))\; dx\\
 \frac 1n\sum_{i=1}^n G(i/n) \mc E_i(0) \longrightarrow \int_0^1 G(x)
    \mf u(\tau(0,x),\beta(0,x))\; dx
  \end{split}
\end{equation}
for any continuous compactly supported test function $G\in \mc
C_0(\R)$. We expect the same convergence to happen at the macroscopic time
$t$:
\begin{equation}
  \label{eq:4o}
  \begin{split}
    \frac 1n\sum_{i=1}^n G(i/n) r_i(t) \longrightarrow \int_0^1 G(x)
     r(t,x) \; dx\\
 \frac 1n\sum_{i=1}^n G(i/n) \mc E_i(t) \longrightarrow \int_0^1 G(x)
     u(t,x)\; dx
  \end{split}
\end{equation}
and  the macroscopic evolution for the volume and energy profiles
should follow the system of equations, for  
$ (t,x)\in \R_+\times [0,1]$
\begin{equation}
  \label{eq:diff}
  \begin{split}
    \partial_t r(t,x) &= \frac 1{2\gamma} \partial_{xx}\big[ \bs\tau(r,u) \big]\\
    \partial_t u(t,x) &= \partial_x \Big[\mc D(r,u) \partial_x \big[\beta^{-1}(r,u)\big]\Big] +
    \frac 1{4\gamma} \partial_{xx}  \left[\bs\tau^2(r,u)\right]
  \end{split}
\end{equation}
% \begin{equation}
%   \label{eq:diff}
%     \partial_t 
%     \left(
%     \begin{matrix}
%       r(t,x)\\
%       u(t,x)
%     \end{matrix}
%   \right)
% = \partial_x \left[\mc D(r(t,x),u(t,x)) \; \partial_x  
%     \begin{pmatrix}
%       r(t,x) \\
%       u(t,x)
%     \end{pmatrix}
%   \right]  \end{equation}
  with the following boundary conditions:
\[\left\{\begin{aligned}
\partial_x \big[\bs \tau(r,u)\big] (t,0) &=0 \\
 \partial_x \big[\beta^{-1}(r, u)\big](t,0)&=0\end{aligned}\right. \qquad 
 \left\{\begin{aligned}
 \bs\tau(r(t,1),u(t,1)) & = \bar\tau(t)\\
 \partial_x \big[\beta^{-1}(r,u)\big](t,1) & = 0
\end{aligned}\right.\] 
and initial conditions 
\[\left\{ \begin{aligned}
r(0,x)&= \mf r\big(\tau(0,x),\beta(0,x)\big) \\ 
u(0,x)&=\mf u\big(\tau(0,x),\beta(0,x)\big).\end{aligned}\right.\]
Equation \eqref{eq:diff} can be deduced by linear response theory
(cf. \cite{bo2}) and
the thermal diffusivity $\mc D$ is defined by the corresponding Green-Kubo
formulas. The convergence of the
corresponding Green-Kubo expression is proved in \cite{bo2}. Still a
proof of the hydrodynamic limit \eqref{eq:4} is out of reach with the known techniques.

In the harmonic case $V(r) = r^2/2$, Equation \eqref{eq:4} is proven in
\cite{simon} with periodic boundary conditions, and we will adapt here
that proof in order to deal with the forcing boundary conditions.

\section{The harmonic case}
\label{sec:harmonic-case}

When the interaction potential is harmonic, explicit computations are available, for instance 
\[ \mc G_{\tau,\beta}=\log\left[ \frac{\beta}{2\pi}
  \exp\left(\frac{\tau^2\beta}{2}\right)\right].\] 
The {thermodynamic relations} between the averaged conserved
quantities $\mf r \in \R$ and $\mf u \in (0,+\infty)$, and the
potentials $\tau \in \R$ and $\beta\in (0,+\infty)$ are given
by 
\begin{equation} {\mf u}(\tau,\beta) = \frac{1}{\beta}+ \frac{\tau^2}{2}, \qquad
\mf r(\tau,\beta) =\tau. \label{rel2} 
\end{equation} 
Furthermore the thermal diffusivity turns out to be equal to $\mc D =
(4\gamma)^{-1}$ (cf. \cite{bo2}).

%Notice that  there exists a bijection between the two sets
%$\left\{(\tau,\beta) \in \R^2 \ ; \ \beta>0 \right\}$ and $\left\{
%  (\mf u,\mf r) \in \R^2 \ ; \ \mf u > {\mf r}^2/2 \right\}$, and the
%equations above can be inverted. 

Let $r_0$ and $u_0$ be two continuous initial
profiles on $[0,1]$, and define the solutions $r(t,\cdot)$ and
$u(t,\cdot)$ to the hydrodynamic equation \eqref{eq:diff}, rewritten as
\begin{align}
\label{eq:linear}
\partial_t r(t,x)&=\frac{1}{2\gamma}\partial_{xx}r(t,x)\notag\\
\partial_t u(t,x)&= 
\frac{1}{4\gamma}\partial_{xx}\left[u(t,x)+\frac{r^2(t,x)}{2}\right]
\end{align} 
with the boundary conditions, for $(t,x)\in\R_+\times[0,1]$
\begin{equation}\label{eq:boundary}\left\{\begin{aligned}
\partial_x r(t,0)&=0     \\       
 r(t,1) & = \bar\tau(t) \\               
 r(0,x)&= r_0(x)                  
\end{aligned}\right. \qquad \qquad \left\{\begin{aligned}  
\partial_x u(t,0)&=0\\
\partial_x u(t,1) & = \bar\tau(t) \partial_x r(t,1)\\
 u(0,x)&=u_0(x).\end{aligned}\right.
\end{equation}
 The solutions $u,r$ are smooth when $t >0$ as soon as the initial condition satisfies $u_0 > r_0^2/2$ (the system of partial differential equations is parabolic). 

In this case, the evolution of $r(t,x)$ is autonomous from
$u(t,x)$, therefore we can call $R(t) = \int_0^1 r(t,x) dx$ the total length of the chain at
time $t$, that also does not depend on $u(\cdot,\cdot)$, and write
the boundary conditions for $u(t,x)$ as 
\begin{equation}
  \label{eq:BC-GLOB}
  \frac{d}{dt} \left[\int_0^1 u(t,x) dx\right] = \bar\tau(t) \dot R(t) = \frac d{dt} L(t)
\end{equation}
where $L$ is the work done by the force $\bar\tau$ up to time $t$. 

For a local function $\phi$, we denote by $\theta_i\phi$ the shift
of the function $\phi$: $\theta_i\phi({\bf r},{\bf p})=\phi(\theta_i {\bf
  r},\theta_i {\bf p})$. This is always well defined for $n$
sufficiently large. The main result is the following: 
\begin{theorem}\label{main} We have
  \begin{equation}
    \label{eq:7}
    \lim_{n\to\infty} \frac {\mc H_n(t)}n = 0
  \end{equation}
where
\begin{equation}
  \label{eq:5}
   \mc H_n(t) = \int f_{t}^n \log
   \left(\frac{f_{t}^n}{\phi_t^n}\right)\; d{\bf r}d{\bf p}
\end{equation}
with  \begin{enumerate}[(i)]
\item $f^n_t$ the density of the
configuration of the system at time $t$, 
\item $\phi_t^n$ the density  of the ``corrected'' local Gibbs measure $\nu_{\tau(t,\cdot),\beta(t,\cdot)}^n$ defined as 
\[ d\nu_{\tau(t,\cdot),\beta(t,\cdot)}^n =
\frac{1}{Z(t)}\prod_{i=1}^n 
e^{ -\beta(t,\frac i n)\big(\mc E_i-\tau(t,\frac i
  n)r_i\big)+\frac 1n F(t,\frac i n)\cdot \theta_i h(\bf r,\bf p)}
dr_i dp_i.\] 
\end{enumerate}
Above $Z(t)$ is the partition function, and $F,h$ are explicit
functions given in \eqref{eq:functions}.
\end{theorem}
We denote by $\mu[\cdot ]$ the expectation with respect to the measure $\mu$. Theorem \ref{main} implies the hydrodynamic limits in the following sense:
\begin{corollary}
 Let $G$ be a continuous function on $[0,1]$ and $\varphi$ be a local
 function which satisfies the following property: there exists a
 finite subset $\Lambda \subset \Z$ and a constant $C>0$ such that,
 for all $({\bf r,p}) \in (\R\times\R)^n$, $ \varphi({\bf r,p})
 \leqslant C\left(1+\sum_{i \in \Lambda} \mc E_i\right)$. Then,  
\begin{equation}
 \mu_t^n \left[\left\vert \frac{1}{n} \sum_i G(i/n) \theta_i \varphi -
     \int_{[0,1]} G(x) \ \tilde{\varphi}(u(t,x),r(t,x))  dx
   \right\vert \right] \xrightarrow[n \to \infty]{} 0
 \end{equation}
 where $\tilde{\varphi}$ is the grand-canonical expectation of $\varphi$: in other words, for any $(u, r)$,  \begin{equation} \tilde{\varphi}({u},{r})=\mu_{\tau,\beta}[\varphi]=\int_{(\R\times\R)^\Z} \varphi({\bf r,p})\ d\mu_{\tau,\beta}({\bf r,p})\ . \end{equation} 
\end{corollary}

We prove Theorem \ref{main} in Section \ref{sec:proof-hydr-limit}.

\section{Thermodynamic consequences}
\label{sec:therm-cons}

\subsection{Second principle of thermodynamics}

Let us first compute the increase of the total thermodynamic entropy,
under the macroscopic evolution given by the general equations \eqref{eq:diff}:
\begin{align} 
    \frac d{dt} \int_0^1 S(r(t,x), u(t,x))\; dx &= \int_0^1 \big[
      -\beta \bs \tau \partial_t r + \beta \partial_t u \big] \; dx \notag\\
   & = \int_0^1 \bigg[ \mc D \left(\frac{\partial_x
          \beta}{\beta}\right)^2 + \frac 1{2\gamma} \beta
      \left(\partial_x \bs \tau\right)^2 \bigg] \; dx \ \geqslant \ 0.  \label{eq:2ndP}
  \end{align}
Assume now that we start in equilibrium with a given constant tension
$\tau_0$ and constant inverse temperature $\beta_0$. To these values
correspond a constant profile of length $r(0,x)=\mc L_0$ and of energy
$u(0,x) = u_0$, that constitute the initial conditions for 
\eqref{eq:diff}. The initial thermodynamic entropy is then $S_0 =
S(\mc L_0, u_0)$. 

We now apply a time depending tension $\bar
\tau(t)$, such that $\bar\tau(t) = \tau_1$ for $t\geqslant \bar t$. It is
clear that the solution converges  as $t\to \infty$ to a new global
equilibrium state, with tension $\tau_1$. This final equilibrium state
has total length $\mc L_1$ given by
\begin{equation}
  \label{eq:finL}
  \mc L_1 = \mc L_0 + \frac{1}{2\gamma} \int_0^\infty \partial_x \big[\bs \tau(r,u)\big] (t,1) \; dt, 
\end{equation}
and energy $u_1 = u_0 + W$,
where $W$ is the mechanical work done by the tension $\bar \tau(t)$.    
The total work $W$ can be computed by:
\begin{equation}
  \label{eq:work}
  W = \frac{1}{2\gamma}\int_0^\infty \bar \tau(t) \partial_x \big[\bs \tau(r,u)\big] (t, 1) \; dt. 
\end{equation}
Consequently the thermodynamic entropy of the final equilibrium state
 equals
\begin{equation}
  \label{eq:finS}
  S_1 = S(\mc L_1, u_1) = S_0 +  \int_0^\infty dt \int_0^1 \bigg[ \mc D \left(\frac{\partial_x
          \beta}{\beta}\right)^2 + \frac 1{2\gamma} \beta
      \left(\partial_x \bs \tau\right)^2 \bigg] \; dx. 
\end{equation}
This is in agreement with the second principle of thermodynamics, in
the statement that an irreversible adiabatic transformation increases
the thermodynamic entropy of the system.

In the harmonic case, the thermodynamic entropy is a function of the
 temperature only, and
\begin{equation}
  \label{eq:entvarh}
  S_1 - S_0 = \log \left(\frac{\beta_0}{\beta_1}\right).
\end{equation}
In other words, any increase of entropy implies an increase of temperature. It
means that any adiabatic irreversible transformation can only increase the
temperature of the system. In the harmonic case,  the
reversible transformations obtained by the quasi-static limit cannot
change the entropy and the temperature. 

% {\color{red} \emph{We have to understand better this relation between
%     non-linearity and possible cooling of the system.}}

\subsection{Quasistatic limit}
\label{sec:quasistatic-limit}

Notice that \eqref{rel2} suggests to define \[\beta^{-1}(t,x) = u(t,x) - \frac 12
r^2(t,x).\] Equation \eqref{eq:linear}  can be written as
\begin{align}
\label{eq:linear2}
\partial_t r(t,x)&=\frac{1}{2\gamma}\partial_{xx}r(t,x)\notag\\
\partial_t \big[\beta^{-1}\big](t,x)&= 
\frac{1}{4\gamma}\partial_{xx} \big[\beta^{-1}\big](t,x) + \frac{1}{2\gamma}
\big(\partial_x r(t,x)\big)^2
\end{align} 
with the boundary conditions, for $(t,x)\in\R_+\times[0,1]$
\begin{equation}\label{eq:boundary-lin2}
\left\{\begin{aligned}
\partial_x r(t,0)&=0     \\       
 r(t,1) & = \bar\tau(t) \\               
 r(0,x)&= r_0(x)                  
\end{aligned}\right. \qquad \qquad 
\left\{
\begin{aligned}  
\partial_x \big[\beta^{-1}\big](t,0)&=0 = \partial_x \big[\beta^{-1}\big](t,1)\\
 \beta^{-1}(0,x)&=u_0(x) - \frac {r_0^2(x)}{2}.
\end{aligned}\right.
\end{equation}
Consider the case when the exterior tension $\bar\tau(t)$ is equal to
a value $\bar\tau_1$ for any $t\geqslant t_1$.
It is clear that we have the following convergence to equilibrium:
\begin{equation*}
  \begin{split}
    &r(t,x)\ \mathop{\longrightarrow}_{t\to\infty}\ \bar\tau_1, \\
    &\beta^{-1}(t,x) \mathop{\longrightarrow}_{t\to\infty}\
    \bar\beta_1^{-1} = \int_0^1 \left( u_0(x') - \frac {r_0(x')^2}{2}\right) \; dx' + \frac 1{2\gamma} \int_0^\infty dt
    \int_0^1 \big(\partial_x r(t,x)\big)^2 dx.
  \end{split}
\end{equation*}
Suppose, as above, that we start at equilibrium with tension $\tau_0$
and temperature $\beta_0^{-1}$. This means $r(0,x) = \tau_0$,
$u(0,x) = \beta_0^{-1} - \tau_0^2/2$, and an initial exterior
force $\bar\tau(0) = \tau_0$. 
Then, after the limit $t\to\infty$, we have reached a new equilibrium
with tension $\bar\tau_1$ and a higher temperature 
\[
\beta_1^{-1} = \beta_0^{-1} + \frac 1{2\gamma} \int_0^\infty dt
    \int_0^1 \big(\partial_x r(t,x)\big)^2 dx.
\]
In particular the temperature, and consequently the entropy, always 
increase in this irreversible transformation.

We now consider  the quasi-static limit, where we slow down the
changing of the exterior tension, i.e. we consider the same system \eqref{eq:linear2}, but one of the boundary
conditions (precisely, the second one of \eqref{eq:boundary-lin2}) is changed into $ r(t,1) = \bar\tau(\varepsilon t) $. The corresponding solution is denote by $(r^\ve, u^\ve)$. Then Proposition 3.1 of \cite{olla} can be applied and it follows
that
\begin{equation*}
  \lim_{\ve \to 0}  \int_0^\infty dt
    \int_0^1 \big(\partial_x r^\ve(\varepsilon^{-1} t,x)\big)^2 dx \; =  0
\end{equation*}
and $r^\ve(\varepsilon^{-1}t,x) \to \bar \tau(t)$, for all $(t,x) \in \R_+ \times [0,1]$. Consequently 
\[\big(\beta^\ve(\varepsilon^{-1}t,x)\big)^{-1} \mathop{\longrightarrow}_{\ve\to 0} \beta_0^{-1},
\qquad
u^\ve(\varepsilon^{-1}t,x)  \mathop{\longrightarrow}_{\ve\to 0} \beta_0^{-1} - \frac{\bar\tau^2(t)}2
\]
for all $(t,x) \in \R_+ \times [0,1]$. Similar considerations are valid in the non-linear case.

\section{Proof of the hydrodynamic limit}
\label{sec:proof-hydr-limit}

We approach this problem by using the relative entropy method
\cite{yau}. We adapt the 
proof of \cite{simon}, where the same harmonic perturbed chain is
investigated, assuming periodic boundary conditions. We recall here
the main steps of the argument, and give details only for computations
that change due to boundary conditions.  

In the context of diffusive systems, the relative entropy method works
if the following conditions are satisfied. \begin{enumerate} 
\item First, the dynamics has to be \textit{ergodic}: the only time
  and space invariant measures for the infinite system, with finite
  local entropy, are given by mixtures of  Gibbs measures in infinite
  volume $\mu_{\tau,\beta}$.  From \cite{ffl}, we know that the
  velocity-flip model is ergodic in the sense above. For a precise
  statement, we refer to \cite[Theorem 1.3]{simon}.  
\item Next, we need to establish the so-called \textit{fluctuation-dissipation equations}. Such equations express the microscopic currents $j_i^{\mc E}$ and $j_i^r$ (respectively of energy and deformation) as the sum of a discrete gradient and a fluctuating term.
Here, the conservation laws write for $i \geqslant 1$,
\begin{align*}
\mc L_n^\tau(\mc E_i)&=n^2( j_{i+1}^{\mc E}-j_i^{\mc E}) \ \text{ with } j_i^{\mc E}:= \begin{cases} r_ip_{i-1}, & \text{ if } i\in \{1,\dots n\}, \\ \tau p_n, & \text{ if } i=n+1, \end{cases} \\
\mc L_n^\tau(r_i)&=n^2(j_{i+1}^{r}-j_i^r)  \ \text{ with } j_i^r=p_{i-1} \text{ for any } i \in \{1,\dots, n+1\}.
\end{align*}
Notice that $j_1^{\mc E}=0$ and $j_1^r=0$. If $\tau_if(\bf r,\bf p)$
is a local function on the configurations, we define its discrete
gradient as 
\[ 
\nabla(\theta_if):=\theta_{i+1}f-\theta_if.
\]
We denote by $(\mc L_n^\tau)^\star:=-n^2A_n^\tau+\gamma n^2S_n$ the
adjoint of $\mc L_n^\tau$ in ${\bf L}^2(\mu_{\tau,\beta}^n)$.  
We write down the fluctuation-dissipation equations: for $i\in\{2,\dots, n\}$,
\begin{align}
 j_{i}^{\mc E}&=\nabla(u_i) +  {(\mc L_n^\tau)^\star} \left[-\frac{r_{i}\big(p_{i-1}+p_{i}-\gamma r_{i}\big)}{4\gamma n^2}\right] \label{eq:je}\\
  j_{i}^r& = \nabla\left(-\frac{r_{i-1}}{2\gamma}\right)+(\mc
  L_n^\tau)^\star\left[ -\frac{p_{i-1}}{2\gamma
      n^2}\right] \label{eq:jr} 
\end{align}
where  for $i\in\{2,\dots,n\}$, 
\[
u_i=-\frac{p_{i-1}^2+r_{i-1}r_{i}}{4\gamma}  \quad \text{ and } \quad
u_{n+1}=-\frac{p_n^2+\tau r_n}{4\gamma}. 
\]
For $i=n+1$, the fluctuation-dissipation equations read as
\begin{align*}
 j_{n+1}^{\mc E}&=\tau\left(\frac{r_n - \tau}{2\gamma} +  {(\mc L_n^\tau)^\star} \left[-\frac{p_n}{2\gamma n^2}\right]\right) \\
  j_{n+1}^r& =\frac{r_n - \tau}{2\gamma} +  {(\mc L_n^\tau)^\star} \left[-\frac{p_n}{2\gamma n^2}\right]
\end{align*}
\item Since we observe the system on a diffusive scale and the system is non-gradient, we need second order approximations. If we want to obtain the entropy estimate of order $o(n)$, we can not work directly with the local Gibbs measure $\mu_{\tau(t,\cdot),\beta(t,\cdot)}^n$: we have to correct it with a small term.

\item Finally, we need to control all the following moments, \begin{equation}
\label{eq:mombound}
\int \bigg\{\frac 1n \sum_{i=1}^n |{\mc E}_i|^k\bigg\} d\mu_t^n, \quad k\geqslant 2 
\end{equation} uniformly in time and with respect to $n$. The harmonicity of the chain is crucial to get this result: roughly speaking, it ensures  that the set of mixtures of Gaussian probability measures is left invariant during the time evolution. 

\end{enumerate}

In the two next subsections, we explain the relative entropy method, and highlight the role of the fluctuation-dissipation equations. In Subsection \ref{sec:mom-bounds}, we prove bounds \eqref{eq:mombound}.

\subsection{Relative entropy method}
\label{sec:rel-entropy}

Recall the definition of the relative entropy \eqref{eq:5}. The
objective is to prove a Gronwall estimate of the entropy production in
the form 
\begin{equation} 
\frac{d}{dt} \mc H_n(t) \leqslant C \ \mc H_n(t)+o(n) , \label{gron} 
\end{equation} 
where $C >0$ does not depend on $n$. We begin with the following lemma, proved in \cite[Chap. 6, Lemma 1.4]{kipnis}.

\begin{lemma}\label{entropy}
\begin{equation*}
\frac{d}{dt}\mc H_n(t) \leqslant \int
  \frac{1}{\phi_t^n}\big\{(\mc L_n^{\bar
    \tau(t)})^\star\phi_t^n-\partial_t\phi_t^n\big\} f_t^n \; d{\bf
    r}d{\bf p}= \int  \frac{1}{\phi_t^n}\big\{(\mc
  L_n^{\bar\tau(t)})^\star\phi_t^n-\partial_t\phi_t^n\big\} \;
  d\mu_t^n.
\end{equation*} 
\end{lemma}

We now choose the correction term: for $i\neq n$ let us define 
\begin{equation} 
\left\{ 
\begin{aligned} F\left(t,i/n\right) & :=\big(\partial_x\beta\left(t,i/n\right),  - \partial_x(\tau\beta)\left(t,i/n\right) \big), \\
\theta_i h({\bf r},{\bf p}) & :=
\left(-\frac{r_{i+1}\big(p_{i}+p_{i+1}-\gamma r_{i+1}\big)}{4\gamma},
  -\frac{p_{i}}{2\gamma}\right). 
\end{aligned}\right.  \label{eq:functions} 
\end{equation}
For $i=n$, we assume
 \begin{equation*}
 \left\{ \begin{aligned} F(t,1) & :=\big(0,  (\beta\partial_x\tau)\left(t,1\right) \big), \\
\theta_n h({\bf r},{\bf p}) & := \Big(0,
-\frac{p_{n}}{2\gamma}\Big). \end{aligned}\right.    
\end{equation*}
 For the sake of simplicity, we introduce the following
 notations \[\xi_i:=(\mc E_i,r_i), \quad \chi:=(\tau,\beta), \quad
 \eta(t,x):=(u(t,x),r(t,x)).\] If $f$ is a vectorial function, we
 denote its differential by $Df$. We are now able to state the main
 technical result of the relative entropy method. 
\begin{proposition} \label{prop} The term $ (\phi_t^n)^{-1}\big\{(\mc
  L_n^{\bar\tau(t)})^\star\phi_t^n-\partial_t\phi_t^n\big\}$ is given
  by a finite sum of  microscopic expansions up to the first order.  
In other words,  it can be written as a finite sum, for which each
term $k$ is of the form 
\begin{equation} \sum_{i=1}^n v_k\Big(t,\frac{i}{n}\Big)
  \left[J_i^k-H_k\bigg({\eta}\Big(t,\frac{i}{n}\Big)\bigg)-(DH_k)\bigg({\eta}\Big(t,\frac{i}{n}\Big)\bigg)\cdot
    \bigg({\xi}_i-{\eta}\Big(t,\frac{i}{n}\Big)\bigg)\right] +
  o_t(n)  \label{eq:tay} 
\end{equation} 
where \begin{itemize}
\item $o_t(n)$ is an error term in the sense that \[ \int_0^t ds \int n^{-1} o_s(n) f_s^n \; d{\bf r} \; d{\bf p} \xrightarrow[n\to\infty]{} 0,\]
\item $J_i^k$ are local functions on the configurations given in Subsection \ref{sec:taylor}, 
\item $v_k(t,x)$ are smooth functions that
depends on $\tau,\beta$, given in  Subsection \ref{sec:taylor}, 
\item the functions $H_k$ satisfy
\begin{equation}
 H_k\left({\eta}\Big(t,\frac{i}{n}\Big)\right)=
 \mu_{\chi(t,i/n)}^n\big[J_0^k\big]. \label{eq:mean} 
\end{equation}
\end{itemize}
\end{proposition}

Before explaining the main steps to prove Proposition \ref{prop}, let us achieve the proof of Theorem \ref{main}. A priori  the first term on the right-hand side of \eqref{eq:tay} is of order $n$, but we can  take advantage of these microscopic Taylor expansions. First, we need to cut-off large energies in order to work  with bounded variables only. Second, the strategy consists in performing a one-block estimate: we replace the empirical truncated current which is averaged over a microscopic box centered at $i$  by its mean with respect to a Gibbs measure with the parameters corresponding
to the microscopic averaged profiles. This is achieved thanks to the ergodicity of the dynamics. A one-block estimate is performed for each term of the form  \begin{equation*} 
\sum_{i=1}^n v_k\Big(t,\frac{i}{n}\Big) \left[J_i^k-H_k\left({\eta}\Big(t,\frac{i}{n}\Big)\right)-(DH_k)\left({\eta}\Big(t,\frac{i}{n}\Big)\right)\cdot\left({\xi}_i-{\eta}\Big(t,\frac{i}{n}\Big)\right)\right].
\end{equation*} 
We deal with error terms by taking advantage of \eqref{eq:mean} and by using the large deviation properties of the probability measure $\nu_{\chi(t,\cdot)}^n$, that locally is almost homogeneous. Along the proof, we will need to control, uniformly in $n$, the quantity \begin{equation*} \int \sum_{i=1}^n \exp\left(\frac{\mc E_i}{n}\right) \ d\mu_t^n. \end{equation*} In fact, to get the convenient estimate, it is not difficult to see that it is sufficient to prove \eqref{eq:mombound}. 
The rest of the proof follows by the standard arguments of the relative entropy method (cf. \cite{kipnis,oe,  ovy,simon,yau}).

\subsection{Taylor expansion}\label{sec:taylor}

First, let us give the explicit expressions for all the functions given in Proposition \ref{prop}. For $i=1,...,n-1$, we have:

\[\begin{array}{| c | c | c | c |}
\hline k & J_i^k & H_k(u,r) & v_k(t,{x}) \\
\hline \hline \ds 1 & \ds p_i^2+r_ir_{i+1}+2\gamma r_i p_{i-1} & \ds u+\frac{{r}^2}{2} & \ds -\frac{1}{4\gamma} \partial_{xx} \beta(t,{x}) \\ 
2 & r_i+\gamma p_{i-1} & {r} & \ds \frac{1}{2\gamma} \partial_{xx}(\tau\beta)(t,{x}) \\
3 & p_i^2\ (r_i+r_{i+1})^2 & \ds (2u-{r}^2)  \bigg(u+\frac{3}{2}{r}^2\bigg) &\ds  \frac{1}{8\gamma} [\partial_x \beta(t,{x})]^2 \\
4 & p_i^2 \ (r_i+r_{i+1}) & {r} \ (2u-{r}^2)& \ds -\frac{1}{2\gamma} \partial_x\beta(t,{x}) \ \partial_x(\tau\beta)(t,{x}) \\
5 & p_i^2 & \ds u-\frac{{r}^2}{2} & \ds \frac{1}{2\gamma} [\partial_x (\tau\beta)(t,{x})]^2 \\ \hline \end{array}\]
For $i=n$, the local functions $J_n^k$ read:
\[J_n^1=p_n^2+\tau r_n, \quad J_n^2=r_n, \quad J_n^3=J_n^4=0, \quad J_n^5=p_n^2\] associated to 
\[ v_1=-\frac1{4\gamma}\partial_{xx}\beta, \quad v_2=\frac1{2\gamma}\partial_{xx}(\tau\beta), \quad v_5= \frac1{2\gamma} (\beta \partial_x\tau)^2. \]
The  fluctuation-dissipation equations are crucial: the role of functions $F,h$ is to compensate the fluctuating terms. For the sake of clarity, we write down three different lemmas. Let us introduce the notation, for $i \in \{1,\dots,n\}$,
\[
\delta_i({\bf r},{\bf p})=  F\left(t,i/n\right) \cdot \theta_ih({\bf
  r,p}),
\]
 where we denote by $a\cdot b$ the usual scalar product in $\bb R^2$.

\begin{lemma}[Antisymmetric part]\label{lem:antisym}

\begin{align}
n^2 A_n^{\bar\tau(t)}\phi_t^n =& {\phi_t^n} \sum_{i=0}^{n-1}
\left\{\partial_{xx}\beta\Big(t,\frac{i}{n}\Big)
  \left[\frac{r_{i+1}p_{i}}{2}-u_{i+2}\right]
  -\partial_{xx}(\beta\tau)\Big(t,\frac{i}{n}\Big)\left[\frac{p_{i}}{2}+\frac{r_{i+1}}{2\gamma}\right]
\right\}\notag\\ 
& + \phi_t^n n \sum_{i=1}^{n-1} \Big\{(n^2\mc
L_n^{\bar\tau(t)})^\star(\delta_i) + A_n^{\bar\tau(t)}(\delta_i)\Big\}
+  n\ \frac{\phi_t^n}{2\gamma} (\tau\beta\partial_x\tau)(t,1) +
o(n). \notag\end{align} 
\end{lemma}

\begin{proof}
The first step consists in performing an integration by part coming from the conservation laws. One can easily check that 
\begin{align*}
n^2 A_n^{\bar\tau(t)}\phi_t^n= & {\phi_t^n}\sum_{i=1}^{n-1} n
\left[ \partial_x\beta\Big(t,\frac{i}{n}\Big)j_{i+1}^{\mc E}
  - \partial_x(\beta\tau)\Big(t,\frac{i}{n}\Big)j_{i+1}^r\right] \\ 
&+ {\phi_t^n}\sum_{i=1}^{n-1}
\frac{1}{2}\left[ \partial_{xx}\beta\Big(t,\frac{i}{n}\Big)j_{i+1}^{\mc E}
  - \partial_{xx}(\beta\tau)\Big(t,\frac{i}{n}\Big)j_{i+1}^r\right]
+o\left(n\right) \\ 
& + \phi_t^n n \sum_{i=1}^{n} A_n^{\bar\tau(t)}(\delta_i)+
n^2\Big((\beta\tau)(t,1)p_n-\beta(t,1)\bar\tau(t)p_n\Big). 
\end{align*}
Note that the boundary conditions $\partial_x\beta(t,0)=0$ and $\partial_x(\tau\beta)(t,0)=0$ permit to introduce the boundary gradients. Moreover, the condition $\tau(t,1)=\bar\tau(t)$ makes the last two terms compensate. 

The next step makes use of the fluctuation-dissipation equations. The
fluctuating terms in the range of $(\mc L_n^{\bar\tau(t)})^\star$ give
the contribution $\sum (\mc L_n^{\bar\tau(t)})^\star(\delta_i)$  (for
$i=1,..., n-1$) whereas the gradient terms are turned into a second
integration by parts. The term $A_n^{\bar\tau(t)}(\delta_n)$ is going
to be treated separately.  Then, one can check that  
\begin{align*}
n^2 A_n^{\bar\tau(t)}\phi_t^n =& {\phi_t^n} \sum_{i=0}^{n-1}
\left\{\partial_{xx}\beta\Big(t,\frac{i}{n}\Big)
  \left[\frac{r_{i+1}p_{i}}{2}-u_{i+2}\right]
  -\partial_{xx}(\beta\tau)\Big(t,\frac{i}{n}\Big)\left[\frac{p_{i}}{2}+\frac{r_{i+1}}{2\gamma}\right]
\right\}\notag\\ 
& + n\phi_t^n  \sum_{i=1}^{n-1} \Big\{(
n^{-2}\mc L_n^{\bar\tau(t)})^\star(\delta_i) + A_n^{\bar\tau(t)}(\delta_i)\Big\}
+ o\left({n}\right) \notag\\ 
& + n{\phi_t^n}\left[-  \partial_x\beta(t,1)
  \frac{p_n^2+\bar\tau(t)r_n}{4\gamma} + \partial_x(\tau\beta)(t,1)
  \frac{r_n}{2\gamma}+ A_n^{\bar\tau(t)}(\delta_n)\right].
 \end{align*}
Remind that $\partial_x\beta(t,1)=0$. After simplifications in the last line above, we get
\begin{align*}
n^2 A_n^{\bar\tau(t)}\phi_t^n = &  {\phi_t^n} \sum_{i=0}^{n-1}
\left\{\partial_{xx}\beta\Big(t,\frac{i}{n}\Big)
   \left[\frac{r_{i+1}p_{i}}{2}-u_{i+2}\right]
  -\partial_{xx}(\beta\tau)\Big(t,\frac{i}{n}\Big)\left[\frac{p_{i}}{2}+\frac{r_{i+1}}{2\gamma}\right]
\right\}\notag\\ 
& + n {\phi_t^n} \sum_{i=1}^{n-1} \Big\{(n^{-2} \mc
L_n^{\bar\tau(t)})^\star(\delta_i) + A_n^{\bar\tau(t)}(\delta_i)\Big\}
+  n\ \frac{\phi_t^n}{2\gamma}   (\tau\beta\partial_x\tau)(t,1) +
o\left({n}\right).  
\end{align*}

\end{proof}

The following lemma is widely inspired from \cite{simon}. As previously, we keep the term $S_n(\delta_n)=-2\gamma\delta_n$  isolated.

\begin{lemma}[Symmetric part]\label{lem:sym}
\begin{equation*}
\frac{ n^2S_n (\phi_t^n)}{\phi_t^n} = {n} \sum_{i=1}^{n-1} S_n(\delta_i) + 
n (\beta\partial_x\tau)(t,1) p_n + \frac{1}{4}
\sum_{y=1}^{n} \left( \sum_{i=1}^{n} \delta_i({\bf p}^y)-\delta_i({\bf
    p})\right)^2 + \ \varepsilon(n), 
\end{equation*} 
where
$\displaystyle \mu_t^n\left[\varepsilon(n) \right] = o(n)$. 
\end{lemma}

The proof of Lemma \ref{lem:sym} is the same as in \cite[Lemma A.2]{simon}, provided that moment bounds have been proved (see Section \ref{sec:mom-bounds}).
%More precisely, in \eqref{taylor} we need to bound terms that come from the symmetric
%part of the generator that are in the form
%\[\frac 1{n}\sum_{y=1}^n |r_y|^3 |p_y|^3 e^{\frac cn |r_y| |p_y|}\]
%where $c>0$ is constant. This is due to the inequality
% $|e^x -1 -x -\frac{x^2}2| \leqslant |x|^3 e^{|x|}$. One can write
%\begin{align*}
%    \frac 1{n}\sum_{y=1}^n |r_y|^3 |p_y|^3 e^{\frac cn |r_y| |p_y|} &\leqslant
%    \frac 1{n}\sum_{y=1}^n |r_y|^3 |p_y|^3 e^{\frac cn \mc E_y} \\
%   & \leqslant \frac C{n}\sum_{y=1}^n |r_y|^3 |p_y|^3 \leqslant
%   \frac C{2n}\sum_{y=1}^n \left(|r_y|^6+ |p_y|^6\right)
% \\ &\le Cn^2 \left( \frac 1{n} \sum_y
%    \left(|r_y|^2+ |p_y|^2\right) \right)^3 \le C E^3 n^2
%  \end{align*}
%{\color{red} I left the last line to show you where I was making a
%  mistake}.
%
%
%since we can provide a uniform bound on the total energy: \[\frac 1n
%\mc E_y \leqslant \frac 1n \sum_x \mc E_x \leqslant E.\] 
%%
 The last result below can also be proved by following straightforwardly \cite{simon}.

\begin{lemma}[Logarithmic derivative]\label{lem:derivative}
\[\partial_t\{\log(\phi_t^n)\}  = \sum_{i=1}^n  -\Big[{\mc E}_i-u\Big(t,\frac{i}{n}\Big)\Big] \partial_t \beta\Big(t,\frac{i}{n}\Big) + \Big[r_i -r\Big(t,\frac{i}{n}\Big)\Big] \partial_t (\tau\beta)\Big(t,\frac{i}{n}\Big) + O(1).\]
\end{lemma}

We are now able to prove the Taylor expansion. According to the three previous results and to the notations introduced at the beginning of Subsection \ref{sec:taylor} we have
\begin{align} \frac{1}{\phi_t^n}& ({\mc L}_n^{\bar\tau(t)})^\star \phi_t^n -\partial_t\{\log(\phi_t^n)\} = \sum_{k=1}^5 \sum_{i=1}^{n}v_k\Big(t,\frac{i}{n}\Big) J_i^k    \notag \\
 & + \sum_{i=1}^n  \left\{\Big[\mc E_i-u \Big(t,\frac{i}{n}\Big)\Big] \partial_t \beta\Big(t,\frac{i}{n}\Big) - \Big[r_i -{r}\Big(t,\frac{i}{n}\Big)\Big] \partial_t (\tau\beta)\Big(t,\frac{i}{n}\Big)\right\} \notag \\
& + n (\beta\partial_x\tau)(t,1)\Big(\frac{\tau(t,1)}{2\gamma}+p_n\Big)+ o(n). \label{part1} \end{align} 
In \eqref{part1}, the two boundary terms are treated in the following way: the first term 
\[  n (\beta\partial_x\tau)(t,1)\frac{\tau(t,1)}{2\gamma}  \]
cancels out with the Taylor expansion (see below), and we are going to prove  in Lemma \ref{lem:conv} that the term
$np_n$ is of order $o(n)$ when integrated with respect to $\mu_t^n$.  Recall that $H_k$ is the function defined as follows: \begin{equation*}H_k\left(\eta\Big(t,\frac{i}{n}\Big)\right)=\mu^n_{\chi(t,i/n)}\left[J_0^k\right]. \end{equation*}  The next step consists in introducing in \eqref{part1} the sum 
\begin{align*}\Sigma_n:=\sum_{i = 1}^n & \left\{-\frac{1}{4\gamma} \right.  \partial_{xx} \beta\Big(t,\frac{i}{n}\Big)  H_1\left(\eta\Big(t,\frac{i}{n}\Big)\right)+\frac{1}{2\gamma} \partial_{xx}(\tau\beta)\Big(t,\frac{i}{n}\Big)  H_2\left(\eta\Big(t,\frac{i}{n}\Big)\right) \notag \\
& + \frac{1}{8\gamma}\left[\partial_{xx} \beta\Big(t,\frac{i}{n}\Big)\right]^2  H_3\left(\eta\Big(t,\frac{i}{n}\Big)\right) - \frac{1}{2\gamma}  \partial_{x} \beta \partial_x (\tau\beta)\Big(t,\frac{i}{n}\Big)  H_4\left(\eta\Big(t,\frac{i}{n}\Big)\right) \notag \\
&  \left.  +  \frac{1}{2\gamma} \left[\partial_x(\tau\beta)\Big(t,\frac{i}{n}\Big)\right]^2   H_5\left(\eta\Big(t,\frac{i}{n}\Big)\right)\right\}.   \end{align*} Here, $\Sigma_n$ is not of order $o(n)$ because of the boundary conditions. We let the reader write the two suitable integrations by part implying the Riemann convergence \begin{equation} \label{eq:conv} \frac1n \left( \Sigma_n - n \frac{(\beta\tau\partial_x\tau)(t,1)}{2\gamma}\right) \xrightarrow[n\to\infty]{} 0.\end{equation}
There is one remaining lemma to prove:
\begin{lemma} \label{lem:conv}
Let $\varphi(t)$ a smooth function on $\mathbb R_+$. The following bound holds:
\[ \int_0^t ds \int \varphi(s)\; p_n\; f^n_s \; d{\bf r}d{\bf p} \leqslant
\frac Cn \left(\frac 1n +
\int_0^t \mc H_n(s) ds + \mc H_n(t) + \mc H_n(0)\right)  \]
for some positive constant $C$ independent of $n$.

\end{lemma}

\begin{proof}
Since $\frac d{dt} \sum_{i=1}^n r_i(t) = n^2 p_n(t)$, we have:
\[
\int_0^t  \varphi(s)\; p_n(s)\; ds = -\frac 1{n^2} \int_0^t
\varphi'(s) \sum_{i=1}^n r_i(s) \; ds + \frac 1{n^2} \varphi(t)
\sum_{i=1}^n r_i(t) - \frac 1{n^2} \varphi(0)
\sum_{i=1}^n r_i(0).
\]
Recall the {entropy inequality}: for any $\alpha >0$ and any
positive measurable function $F$  we have 
\begin{equation} 
\int F \ d\mu \leqslant \frac{1}{\alpha} \left\{ \log \left(\int
    e^{\alpha F} \ d\nu \right) + \mc H(\mu \vert \nu)
\right\}, \label{entropin}
\end{equation} 
where $\mc H(\mu|\nu)$ is the relative entropy of $\mu$ with respect to
$\nu$. Therefore,
\begin{equation*}
  \int \frac 1{n^2}\sum_{i=1}^n r_i \; f^n_s \; d{\bf
    r}d{\bf p} \leqslant \frac1{\alpha n} \log
  \int \exp\left(\frac {\alpha}{n} \sum_{i=1}^n r_i\right) \phi^n_s \; d{\bf
    r}d{\bf p} +  \frac1{\alpha n} \mc H_n(s)
\end{equation*}
and it is easy to see that the first term of
the right-hand side of the above bound is bounded by $C n^{-2}$ for some constant $C>0$. 
\end{proof}

Eventually, further computations give
\begin{multline}-\frac{ \partial_{xx} \beta}{4\gamma} \ \partial_{u} H_1 + \frac{ \partial_{xx} (\tau\beta)}{2\gamma} \ \partial_{u} H_2  + \frac{ \left[\partial_x\beta \right]^2}{8\gamma} \ \partial_{u} H_3 - \frac{ \partial_x\beta \partial_x(\tau\beta)}{2\gamma} \  \partial_{u} H_4  \\+ \frac{\left[\partial_x(\tau\beta) \right]^2}{2\gamma}  \ \partial_{u} H_5  = -\partial_t \beta, \label{eq:comp1}\end{multline} 
and
\begin{multline}-\frac{ \partial_{xx} \beta}{4\gamma} \  \partial_{{r}} H_1 + \frac{ \partial_{xx}(\tau\beta)}{2\gamma} \ \partial_{{r}} H_2 + \frac{ \left[\partial_x\beta \right]^2}{8\gamma} \ \partial_{{r}} H_3 + \frac{ \partial_x\beta \partial_x(\tau\beta)}{2\gamma} \  \partial_{{r}} H_4  \\+ \frac{\left[\partial_x\tau\beta \right]^2}{2\gamma}  \ \partial_{{r}} H_5  = -\partial_t (\tau\beta) . \label{eq:comp2}\end{multline}
It remains to rewrite \eqref{part1} after introducing $\Sigma_n$, and making a suitable use of \eqref{eq:comp1}, \eqref{eq:comp2} and \eqref{eq:conv}. Eventually, Proposition \ref{prop} is proven.

\subsection{Moment bounds}
\label{sec:mom-bounds}

In this last part we are going to control all the energy moments. The precise statement is the following:

\begin{theorem}
\label{theo:moments}
For every positive integer $k \geqslant 1$, there exists a positive constant $C$ which does not depend on $n$ (but depends on $k$), such that \begin{equation}\mu_t^n \left[\sum_{i=1}^n {\mc E}_i^{k} \right] \leqslant  C \times n. \label{mom} \end{equation}
\end{theorem}
The dependence on $k$ could be precised: we refer the interested reader to \cite{simon}.
The first two bounds ($k=1,2$) would be sufficient to justify the
cut-off of currents, but here we need more bounds  because of the Taylor expansion (Proposition \ref{prop}). 
Since the chain is harmonic, Gibbs states are Gaussian. Remarkably,
all Gaussian moments can be expressed in terms of variances and
covariances.  We start with a graphical representation of the dynamics
of the process given by the generator $\mc L_n^{\bar\tau(t)} /
n^2$. Notice that time is not accelerated in the diffusive scale. To avoid any confusion, the law of this new process  is  denoted  by  ${\nu}_t^n$. Then, we recover the diffusive time accelerated process by: \[ \mu_t^n= \nu^n_{tn^2}. \] 
In the following, we always respect the decomposition of the space
$\R^n \times \R^n$, where the first $n$ components stand for ${\bf r}$
and the last $n$ components stand for ${\bf p}$. All vectors and
matrices are written according to this decomposition. 

Let $\nu$ be a
measure on $\R^n\times \R^n$. We denote by ${\bf m} \in \R^{2n}$ its mean
vector and by ${\bf C} \in \mathfrak{M}_{2n}(\R)$ its covariance
matrix. There exist $\rho:=\nu[{\bf r}] \in \R^n $ , $\pi: =\nu[{\bf
  p}] \in \R^{n}$ and $U, V, Z \in \mathfrak{M}_{n}(\R)$ such that 
  \begin{equation} {\bf m}=(\rho,\pi) \in \R^{2n} \quad \text{ and } \quad
    {\bf C}=\begin{pmatrix} U & ^{\bf t}Z \\ Z & V \end{pmatrix} \in
    \mathfrak{S}_{2n}(\R). \label{eq:defmc}
 \end{equation}  
Hereafter, we denote by $^{\bf t}Z$ the real transpose of the matrix $Z$.
Thanks to a trivial convexity inequality, instead of proving
\eqref{mom} we are going to show 
\begin{equation} \nu_{t}^n
  \bigg[\sum_{i=1}^n p_i^{2k} \bigg] \leqslant C \times n \quad
  \text{ and } \quad \nu_{t}^n \bigg[\sum_{i=1}^n r_i^{2k} \bigg]
  \leqslant C \times n, \label{eq:boundspr} 
\end{equation} 
where $C$ is a constant that does not depend on $t$ nor on $n$.

\begin{proof}  

\textit{(i) Poisson Process and Gaussian Measures -- } We start by giving a graphical representation of the process, based on the Harris description. Let us define the antisymmetric $(2n,2n)$-matrix, written by blocks as 

\[ A:= \begin{pmatrix}
0_n & {\mf A}_n \\
\\
-^{\bf t}{\mf A}_n & 0_n
\end{pmatrix} \quad \text{ where } \quad {\mf A}_n:=\begin{pmatrix} 
 1  &           &            & (0) \\
 -1 & \ddots &  &   \\
   & \ddots & \ddots  &  \\
 (0) &  & -1 & 1  \end{pmatrix} \in \mathfrak{M}_{n}(\R). \]
 Above $0_n$ is the null $(n,n)$-matrix. We also define the $n$-vector 
 \begin{equation*}
 b(t):=\begin{pmatrix}
 0\\
 \vdots \\
 0 \\
 \bar\tau(t)
 \end{pmatrix}.
 \end{equation*}
Let $(N_i)_{i=1...n}$ be a sequence of independent standard Poisson processes of intensity $\gamma$. At time 0 the process  has an initial state $({\bf r,p})(0)$. Let \[T_1=\inf_{t \geqslant 0} \Big\{\text{ there exists }  i\in\{1,\dots,n\} \text{ such that }\ N_i(t)=1 \Big\}\] and $i_1$ the site where the infimum is achieved. During the interval $[0,T_1)$, the process (not accelerated in time) follows the deterministic evolution given by the generator $A_n^{\bar\tau(t)}$. More precisely, during the time interval $[0,T_1)$, $({\bf r,p})(t)$ follows the evolution given by the system: \begin{equation} \frac{dy}{dt}=A \cdot y(t) + b(t). \label{eq:ode}\end{equation}
At time $T_1$, the momentum $p_{i_1}$ is flipped, and gives a new configuration. Then, the system starts again with the deterministic evolution up to the time of the next flip, and so on. Let $\xi:=(i_1, T_1), \dots, (i_q,T_q), \dots$ be the sequence of sites and ordered times for which we have a flip, and let us denote its law by $\mathbb{P}$. Conditionally to $\xi$, the evolution is deterministic, and the state of the process $({\bf r,p})^{\xi}(t)$ is given for all $ t \in [T_q,T_{q+1})$ by 
\begin{equation} ({\bf r,p})^{\xi}(t)=e^{(t-T_q)A} \circ F_{i_q} \circ e^{(T_q-T_{q-1})A} \circ F_{i_{q-1}} \circ \cdots \circ e^{T_1 A} ({\bf r,p})(0) + \Omega^\xi(t)  \label{omega} \end{equation}
where \begin{itemize} \item $F_i$ is the map $({\bf r,p}) \mapsto ({\bf r,p}^i)$.
\item $\Omega^\xi(t)$ is a vector that depends only on $A$, $b(t)$ and $\xi$, and can be written as 
\begin{multline*} \Omega^\xi(t)= \sum_{\ell=0}^{q-1} e^{(t-T_q)A}\circ F_{i_q} \circ e^{(T_q-T_{q-1})A}\circ \cdots \circ F_{i_{\ell+1}}\circ e^{(T_{\ell+1}-T_\ell)A}\int_{T_\ell}^{T_{\ell+1}} e^{-uA}b(u)\ du \\ 
+  e^{(t-T_q)A} \int_{T_q}^{t} e^{-uA}b(u)\ du  .
\end{multline*}
\end{itemize}

If initially the process starts from $({\bf r,p})(0)$ which is distributed according to a Gaussian measure $\nu_0^n$, then $(\bf{r,p})^{\xi}(t)$ is distributed according to a Gaussian measure $\nu^{\xi}_t$. Finally,  the density ${\nu}_t^n$ is  given by the convex combination \begin{equation} {\nu}^n_t(\cdot)=\int {\nu}^{\xi}_t(\cdot) \ d{\bb P}(\xi).\label{law1} \end{equation}
Moreover, we are  able to write the evolution of the mean vector ${\bf m}_t^\xi$ and the covariance matrix ${\bf C}_t^\xi$ of $\nu_t^\xi$. During the interval $[0,T_1)$, ${\bf m}_t$ follows the evolution given by system \eqref{eq:ode}. At time $T_1$, the component $m_{i_1+n}=\pi_{i_1}$ (which corresponds to the mean of $p_{i_1}$) is flipped, and gives a new mean vector. Then, the deterministic evolution goes on up to the time of the next flip, and so on. 

In the same way, during the interval $[0,T_1)$, ${\bf C}_t$ follows the evolution given by the (matrix) system: \begin{equation}\frac{dM}{dt}=AM(t)-M(t)A. \label{eq:diff2}\end{equation} 
At time $T_1$, all the components $C_{i_1+n,j}$ and $C_{i,i_1+n}$ when $i,j \neq i_1+n$  are flipped and the matrix ${\bf C}_{T_1}$ becomes $\Sigma_{i_1} \cdot {\bf C}_{T_1} \cdot\; ^{\bf t}\Sigma_{i_1}$, where $\Sigma_i$ is defined as 
\[\Sigma_i:=\begin{pmatrix} I_n & 0_n \\ 0_n & I_n-2E_{i,i} \end{pmatrix},\]  and so on up to the next flip. Above, $I_n$ is the $(n,n)$-identity matrix, and $E_{i,i}$ is the $(n,n)$-matrix composed by the elements $(\delta_{i,k} \delta_{i,\ell})_{1\leqslant k,\ell\leqslant n}$ where $\delta_{i,k}$ is the Kronecker delta function.
More precisely,
\begin{equation}  {\bf C}^{\xi}_t=e^{(t-T_q)A} \cdot \Sigma_{i_q}   \cdots \Sigma_{i_1} \cdot e^{T_1A} \cdot {\bf C}_0 \cdot e^{-T_1A} \cdot\;  ^{\bf t}\Sigma_{i_1} \cdots \;  ^{\bf t}\Sigma_{i_q} e^{-(t-T_q)A}. \label{eq:Cxi}\end{equation} 
Finally,  the density ${\nu}_t^n$ is  equal to  \begin{equation} {\nu}^n_t(\cdot)=\int {\nu}^{\xi}_t(\cdot) \ d{\bb P}(\xi)=\int G_{\bf m,C}(\cdot) \ d\theta_{{\bf m}_0,{\bf C}_0}^t({\bf m},{\bf C}),\label{law} \end{equation}
where $G_{\bf m,C}(\cdot)$ denotes the Gaussian measure on $(\R\times\R)^n$ with mean ${\bf m}$ and covariance matrix ${\bf C}$, and $\theta_{{\bf m}_0,{\bf C}_0}^t(\cdot,\cdot)$ is the law of the random variable $({\bf m}_t,{\bf C}_t)$, knowing that the Markov process $({\bf m}_t,{\bf C}_t)_{t\geqslant 0}$ described by the graphical representation above starts from $({\bf m}_0,{\bf C}_0)$. We denote by ${\bb P}_{{\bf m}_0,{\bf C}_0}$ the law of the Markov process $({\bf m}_t,{\bf C}_t)_{t\geqslant 0}$, and by ${\bb E}_{{\bf m}_0,{\bf C}_0}$ the corresponding expectation.
Observe that we have, from \eqref{law}, \begin{equation*} {\nu}_t^n[p_i] = \int G_{\bf m,C}(p_i) \ d\theta_{{\bf m}_0,{\bf C}_0}^t({\bf m,C})=\int \pi_i \; d\theta_{{\bf m}_0,{\bf C}_0}^t({\bf m,C}), \end{equation*} 
\begin{equation*} {\nu}_t^n[r_i] = \int G_{\bf m,C}(r_i) \ d\theta_{{\bf m}_0,{\bf C}_0}^t({\bf m,C})=\int \rho_i \; d\theta_{{\bf m}_0,{\bf C}_0}^t({\bf m,C}).\end{equation*}  
Notice that we conveniently denote by $G_{\bf m,C}(f)$ the mean of the function $f$ with respect to the Gaussian measure $G_{\bf m,C}$. Therefore, we rewrite \eqref{eq:boundspr} as 
\[ \nu_t^n \bigg[\sum_{i=1}^n p_i^{2k} \bigg] = \int \sum_{i=1}^n G_{\bf m,C}\big(p_i^{2k}+r_i^{2k}\big) \; d\theta_{{\bf m}_0,{\bf C}_0}^t({\bf m},{\bf C}).\]

\textit{(ii) Control in the covariance matrix -- } First, let us focus on $G_{\bf m,C}\big(p_i^{2k}+r_i^{2k}\big)$. Notice that 
\[G_{\bf m,C}\big(p_i^{2k}\big) = G_{\bf m,C}\big( [p_i-\pi_i+\pi_i]^{2k}\big) \leqslant 2^{2k-1} \left\{ G_{\bf m,C}\big( [p_i-\pi_i]^{2k}\big)+ \pi_i^{2k}\right\}.\]
Remarkably, we can express all the centered moments of a Gaussian random variable as functions of the variance only. In other words, there exists a constant $K_k$ that depends on $k$ but not on $n$ such that 
\[G_{\bf m,C}\big( [p_i-\pi_i]^{2k}\big) \leqslant K_k \; G_{\bf m,C}\big( [p_i-\pi_i]^{2}\big)^k = K_k \; \big(C_{i+n,i+n}\big)^k(t).\]
Therefore, after repeating the same argument for $G_{\bf m,C}\big(r_i^{2k}\big)$  we are reduced to control, for any $\xi$, 
\begin{equation}\label{eq:bounds_part1}\sum_{i=1}^{2n}(C_{i,i}^\xi)^k(t)\end{equation} and besides
\begin{equation}
\label{eq:bounds_part2}
\sum_{i=1}^n \pi_i^{2k}(t),\qquad \sum_{i=1}^n \rho_i^{2k}(t).
\end{equation}
In the following we treat separately \eqref{eq:bounds_part1} and \eqref{eq:bounds_part2}.

\bigskip

\textit{(iii) Control of \eqref{eq:bounds_part1} using the trace -- } Let us fix once for all a sequence $\xi$ a sequence of sites and ordered times for which we have a flip. The matrix $C^\xi_t$ is symmetric, hence diagonalizable, and after denoting its eigenvalues by $\lambda_1, ..., \lambda_{2n}$, we can write \begin{equation*}\text{Tr}([C^\xi_t]^k)=\sum_{i=1}^{2n} \lambda_i^k.\end{equation*}
We have now to compare $  \sum_{i} \lambda_i^k$ with $\sum_{i} [C_{i,i}^\xi]^k(t)$. If we denote by $P_t^\xi$ the orthogonal matrix of the eigenvectors of $C^\xi_t$, then we get $C^\xi_t=(P_t^\xi)^{\ast} \cdot D \cdot P_t^\xi$, where $D$ is the diagonal matrix with entries $\lambda_1,..., \lambda_{2n}$. For the sake of simplicity, we denote by $(P_{i,j})$ the components of $P_t^\xi$. Then, \begin{equation*} [C^\xi_{i,i}]^k(t)  =\bigg( \sum_{j,\ell} P^{\ast}_{i,j} D_{j,\ell} P_{\ell,i}\bigg)^k  = \bigg( \sum_j P_{i,j}^{\ast} \lambda_j P_{j,i} \bigg)^k = \bigg(\sum_{j} P_{i,j}^{\ast} P_{j,i} \cdot \lambda_j\bigg)^k.\end{equation*}
Since $P$ is an orthogonal matrix, $ \sum_{j} P_{i,j}^{\ast} P_{j,i}=1$.
Consequently, we can use the convexity inequality, and we obtain 
\begin{equation*} \sum_{i} [C^\xi_{i,i}]^k (t) \leqslant \sum_i \sum_j P_{i,j}^{\ast} P_{j,i}  \lambda_j^k \leqslant \sum_j \lambda_j^k = \text{Tr}([C^\xi_t]^k).\end{equation*}
 Since $C_0$ and $C^{\xi}_t$ are similar, we have: \[ \text{Tr}([C^{\xi}_t]^k)=\text{Tr}(C^k_0)= \sum_{i=1}^n \frac{1}{\beta_0^{k}(i/n)} + \left(\frac{1}{\beta_0(i/n)}+\tau_0^2(i/n)\right)^{k} \leqslant K'_1n,\] for some constant $K'_1>0$. 
Therefore, the same inequality holds for $ \sum_i [C_{i,i}^\xi]^k(t)$. 
\bigskip

\textit{(iv) Control of \eqref{eq:bounds_part2} -- } For this last paragraph we go back to the diffusive time scale, namely we are going to bound the two quantities
\[\sum_{i=1}^n \pi_i^{2k}(tn^2) \quad \text{ and } \quad \sum_{i=1}^n \rho_i^{2k}(tn^2).\]
Notice that the sequences $\{\pi_i(t)\}_{i}$ and $\{\rho_i(t)\}_i$ satisfy the following system of differential equations: for $i=1,\dots, n$ and $t\geqslant 0$, \[\left\{\begin{aligned} \pi_i' & = \rho_{i+1}-\rho_{i}-2\gamma \; \pi_i,  \\  \rho_i'&=\pi_{i}-\pi_{i-1}, \end{aligned}\right. \quad \text{with } \quad \left\{\begin{aligned} \rho_{n+1}(t)&=\bar \tau(t/n^2), \\ \pi_0(t)&=0. \end{aligned}\right. \]
Let us recenter $\tilde \rho_i(t) = \rho_i(t) - \bar\tau(t/n^2)$, then
the equations become
\[
\left\{\begin{aligned} 
\pi_i' & = \tilde\rho_{i+1}- \tilde\rho_{i} - 2\gamma \;
    \pi_i,  \\  \tilde\rho_i'&=\pi_{i} - \pi_{i-1} -
    \bar\tau'(t/n^2)n^{-2}, 
\end{aligned}\right. \quad
\text{with } \quad \left\{\begin{aligned} \tilde\rho_{n+1}(t)&= 0,
    \\ 
\pi_0(t)&=0. \end{aligned}\right. 
\] 
 Denote by $\Pi$ the column vector $^{\bf t}(\pi_1,\dots,\pi_n,\pi_1',\dots, \pi_n').$ It is not difficult to see that $\Pi(t)$ follows a first order ordinary differential equation written as \begin{equation} \frac{dy}{dt}=M^\pi\cdot y(t)+T^\pi(t),\label{eq:system1}\end{equation} where $M^\pi$ is the following constant block matrix: 
 \[ M^\pi:=\begin{pmatrix}
 0_n & I_n \\
 \\
 D^{\pi} & -2\gamma I_n
 \end{pmatrix} \text{ where } D^\pi:=\begin{pmatrix}
 -2 & 1 &&  & (0) \\
 1 & -2 & 1 & & \\
  & \ddots & \ddots & \ddots &  \\
 &   & 1 & -2 & 1 \\
 (0) & &  & 1 & -1 
 \end{pmatrix}.\]
Above $I_n$ is the $(n,n)$-identity matrix, and the vector $T^\pi(t)$ is the $(2n)$-vector 
\[ T^\pi(t) := \; ^{\bf t} \Big(\underbrace{0, \dots, 0}_{2n-1}, \;\bar\tau'(t/n^2)n^{-2}\;\Big). \]
In the same way, denote by $R$ the column vector $^{\bf t}(\tilde\rho_1,\dots,\tilde\rho_n,\tilde\rho_1',\dots, \tilde\rho_n').$ It is not difficult to see that $R(t)$ follows a first order ordinary differential equation written as \begin{equation} \label{eq:system2}\frac{dy}{dt}=M^\rho\cdot y(t)+T^\rho(t),\end{equation} where $M^\rho$ is the following constant block matrix: 
 \[ M^\rho:=\begin{pmatrix}
 0_n & I_n \\
 \\
 D^{\rho} & -2\gamma I_n
 \end{pmatrix} \text{ where } D^\rho:=\begin{pmatrix}
 -1 & 1 &  &  & (0) \\
  1 & -2 & 1 &  &  \\
  & \ddots & \ddots & \ddots &  \\
  &   & 1 & -2 & 1 \\
 (0) &  &  & 1 & -2 
 \end{pmatrix}\] and $T^\rho(t)$ is the $(2n)$-vector 
\begin{multline*} T^\rho(t) := \; ^{\bf t}\Big(\underbrace{0,\dots, 0}_{2n-1}, \;\bar\tau(t/n^2)\;\Big) \\ -\big[\bar\tau''(t/n^2)n^{-4}+2\gamma \bar\tau'(t/n^2)n^{-2}\big] \times \; ^{\bf t}\Big(\underbrace{0,\dots,0}_n,\underbrace{1,\dots,1}_n\Big). \end{multline*}
Both matrices $D^\pi$ and $D^\rho$ represents the discrete Laplacian operator with mixed Dirichlet-Neumann boundary conditions. Let us focus on $\Pi(t)$. We are going to compute the characteristic polynomial of $M^\pi$, that is  $\chi^\pi(X):=\det(XI_{2n}-M^\pi)$. One can easily check that \[\chi^\pi(X)=\det(D^\pi-X(X+2\gamma)I_n).\]
 In other words, the eigenvalues of $M^\pi$ are exactly equal to the solutions of \[ x(x+2\gamma)=-\lambda,\] where $-\lambda$ takes any eigenvalue of $D^\pi$. It is well-known that the eigenvalues of $D^\pi$ are all negatives. Therefore, we need to solve $x(x+2\gamma)+\lambda=0$, where $\lambda$ is positive. Precisely, \begin{enumerate}[(i)]
 \item if $\gamma^2 > \lambda$, then the two solutions are real negative numbers written as \[x_{\pm}=-\gamma \pm \sqrt{\gamma^2-\lambda}<0,\]
 \item if $\gamma^2 < \lambda$, then the two solutions are complex numbers written as \[x_{\pm}=-\gamma \pm i\sqrt{-\gamma^2+\lambda},\]
 \item if $\gamma^2 = \lambda$, then $-\gamma$ is the unique solution.
 \end{enumerate}
 As a consequence, every eigenvalue of $M^\pi$ has a negative real part, and the system \eqref{eq:system1} is hyperbolic (and the same holds for $M^\rho$). Let us write the solution of system \eqref{eq:system1} at time $tn^2$: \[\Pi(tn^2)=\exp(tn^2\;M^\pi)\; \Pi(0) + \int_0^{tn^2} \exp((tn^2-s)\;M^\pi)T^\pi(s)\; ds.\]
We are interested in the quantity $\sum_i |\pi_i(tn^2)|^m$, which is less or equal than the following norm 
\[ \Big(\big\Vert\Pi(tn^2)\big\Vert_m\Big)^m:= \sum_{i=1}^n \Big\{|\pi_i(tn^2)|^m + |\pi_i'(tn^2)|^m \Big\}.\] Since the system is hyperbolic, there exists a constant $C>0$ such that, for every $s \in [0,t]$,  
\[ \big\Vert\exp((tn^2-s)\; M^\pi)\; \Pi(0)\big\Vert_m \leqslant C \big\Vert \Pi(0) \big\Vert_m.\] Observe that the initial condition writes 
\[ \big\Vert \Pi(0) \big\Vert_m^m =\sum_{j=1}^{n-1} \Big| \tau_0\Big(\frac{j+1}{n}\Big) -  \tau_0\Big(\frac{j}{n}\Big)\Big|^m + \big|\bar\tau(0)-\tau_0(1)\big|^m.\]
The last term above vanishes due to the assumptions on the boundary  \eqref{eq:boundary}. Since the profile $\tau_0$
  is smooth, it is clear that $\Vert \Pi(0) \Vert_m^m$ is of order $n^{1-m}$. On the other hand,
  \begin{align*}\left\Vert \int_0^{tn^2}  \exp((tn^2-s)\; M^\pi)T^\pi(s)\; ds \right\Vert_m^m &\leqslant C^m\left( \int_0^{tn^2} \big\Vert T^\pi(s) \big\Vert_m \; ds \right)^m \\
  &=\left(\int_0^{tn^2} n^{-2}\Big|\bar\tau'\Big(\frac{s}{n^2}\Big)\Big| \; ds\right)^m \\
  & = \left(\int_0^t |\bar\tau'(u)|\; du\right)^m \end{align*}
  so that the bound does not depend on $n$. Therefore, we proved that there exists a constant $K'_2$ that does not depend on $n$ nor on $t$ such that 
  \[\sum_{i=1}^n |\pi_i(tn^2)|^m \leqslant \big\Vert \Pi(tn^2) \big\Vert_m^m \leqslant K'_2 \; n.\]
  The same argument is valid for $R(t)$, except two different estimates: the first one appears in the initial condition, which now reads 
  \[ \big\Vert R(0) \big\Vert_m^m =\sum_{j=1}^{n} \Big| \tau_0\Big(\frac{j}{n}\Big)-\bar\tau(0)\Big|^m + \sum_{j=1}^{n} \big|\bar\tau'(0)n^{-2}\big|^m .\]
Hence, $\Vert R(0) \big\Vert_m^m$ is of order $n$ (instead of $n^{1-m}$), but this is enough. The second difference comes from the vector $T^\rho(t)$. Now we have to control
\[\left(\int_0^{tn^2} \bigg[\Big|\bar\tau\Big(\frac{s}{n^2}\Big)\Big|^m + n^m\Big|\bar\tau''\Big(\frac{s}{n^2}\Big)n^{-4}+\bar\tau'\Big(\frac{s}{n^2}\Big)n^{-2}\Big|^m\bigg]^{1/m}  \; ds\right)^m,\]
which is also bounded uniformly in $n$. Therefore, we conclude  that there exists a constant $K'_3$ that does not depend on $n$ such that 
  \[\sum_{i=1}^n \big|\rho_i(tn^2)-\bar\tau(t)\big|^m \leqslant \big\Vert R(tn^2) \big\Vert_m^m \leqslant K'_3 \; n,  \]
which implies
\[  \sum_{i=1}^n \big|\rho_i(tn^2)\big|^m \lesssim K'_3 \; n + \sum_{i=1}^n \big|\bar\tau(t)\big|^m \leqslant K'_4 \; n.\]

\end{proof}

\section*{Acknowledgments}

The authors warmly thank Cedric Bernardin for his useful suggestions on this work.

\bibliographystyle{amsalpha}

\end{document}